\documentclass[11pt, letterpaper]{article}

\usepackage{arxiv}

\usepackage{times}
\usepackage[colorlinks=true, linkcolor=purple]{hyperref}
\usepackage{amsthm}
\usepackage{amsmath,amssymb,amsfonts}
\usepackage{helvet}
\usepackage{graphicx}
\usepackage{enumerate}
\usepackage{pgfplots}
\usepackage{mathtools,relsize}
\usepackage{float}
\usepackage{subfigure}
\usepackage{subcaption}
\usepackage{latexsym}
\usepackage{setspace}
\newtheorem{theorem}{Theorem}
\newtheorem{definition}{Definition}
\newtheorem{lemma}{Lemma}
\newtheorem{corollary}{Corollary}
\newtheorem{fact}{Fact}

\DeclareMathOperator{\interior}{int}

\usepackage{biblatex} 
\begin{document}
\title{Revisit the Arimoto-Blahut algorithm: New Analysis with Approximation}

\author{
    Michail Fasoulakis${}^{\dagger}$ \\ \texttt{Michail.Fasoulakis@rhul.ac.uk} \And Konstantinos Varsos${}^{\ddagger}$ \\ \texttt{varsosk@uoc.gr} \And Apostolos Traganitis${}^{\S}$ \\
    \texttt{tragani@ics.forth.gr} \and \\
     ${}^{\dagger}$Department of Computer Science, \\Royal Holloway, University of London, UK  \\
     ${}^{\ddagger}$Department of Mathematics and Applied Mathematics, \\
    University of Crete, Greece  
    \\
    ${}^{\S}$Institute of Computer Science,\\ 
 Foundation for Research and Technology-Hellas (FORTH)
}

\date{}

\maketitle
\begin{abstract}
By the seminal paper of Claude Shannon \cite{Shannon48}, the computation of the capacity of a discrete memoryless channel has been considered as one of the most important and fundamental problems in Information Theory. Nearly 50 years ago, Arimoto and Blahut independently proposed identical algorithms to solve this problem in their seminal papers \cite{Arimoto1972AnAF, Blahut1972ComputationOC}. The Arimoto-Blahut algorithm was proven to converge to the capacity of the channel as $t \to \infty$, with a convergence rate upper bounded by $O\left(\log(m)/t\right)$, where $m$ is the size of the input distribution. Under the assumption that a unique optimal solution is in the interior of the input probability simplex, the convergence becomes inverse exponential after an iteration $t^0$ \cite{Arimoto1972AnAF}. More recently, it was demonstrated in \cite{Nakagawa2020AnalysisOT} that in certain specific cases, the convergence rate is at worst case inverse linear.\\
In this paper, we revisit this fundamental algorithm analyzing its rate of convergence focusing on the approximation of the capacity. Our main result shows that the convergence rate to an $\varepsilon$-optimal solution, for any sufficiently small constant $\varepsilon > 0$, is inverse exponential $O\left(\log(m)/c^t\right)$, for some constant $c > 1$. Given this, we derive new and complementary results for the computation of capacity, particularly in cases where an exact solution is sought.
\end{abstract}

\section{Introduction}\label{sec:Introduction}

In Information theory \cite{Cover2006}, the Arimoto-Blahut algorithm stands as a cornerstone method for efficiently computing the capacity of discrete memoryless channels and solving related optimization problems. Originally introduced by Suguru Arimoto and Richard Blahut in the 1970s \cite{Arimoto1972AnAF,Blahut1972ComputationOC}, it has been celebrated for its convergence properties and practical applicability in diverse communication systems. In particular, a series of input distributions is computed by a recurrence formula produced by an alternating optimization approach, converging to the input probability distribution that achieves the channel capacity. In the paper of Arimoto (see also \cite{SutterSEL15,Jurgensen1984ANO}), this algorithm was proved to have an upper bound of the rate of convergence equals to $O\Big(\log(m)/{t}\Big)$. However, it is an open question if a better upper bound holds for the rate of convergence.  

Recently, this algorithm has been scrutinized heavily in \cite{Nakagawa2020AnalysisOT}, where the authors perform Taylor expansions and eigenvalues' analysis of a Jacobian matrix to determine the conditions for the convergence rates of the algorithm to an optimal solution, showing that the algorithm in some cases converges with exponential, but in other cases with inverse linear convergence rate under some specific conditions.  Furthermore, in another recent work by \cite{SutterSEL15}, utilizing the dual formulation for computing capacity, an a priori error bound of $O\left({n\sqrt{\log(m)}}/{t}\right)$ has been demonstrated. This results in overall computational complexity of $O\left({mn^2 \sqrt{\log(m)}}/{\varepsilon}\right)$ for finding an $\varepsilon$-solution. 

In this paper, we take a different point of view, observing that when the Arimoto-Blahut algorithm is initialized in the interior of the probability simplex, the sequence of probability distributions it generates converges at an inverse exponential rate to an approximate optimal solution from the interior, for any constant sufficiently small approximation parameter $\varepsilon$. The key intuition behind our approach is to partition the convergence domain into two regions depending on the constant approximation $\varepsilon$ of the capacity. Specifically, we define a region $R$ around the convex set of the optimal solutions, which contains all $\varepsilon$-optimal input distributions, and its complement. This partitioning enables us to prove that for any constant $\varepsilon > 0$, the algorithm converges to region $R$ at an inverse exponential rate. In other words, the Arimoto-Blahut algorithm reaches a constant approximation solution of the capacity exponentially fast for any sufficiently small constant approximation. Furthermore, applying tools derived from this approach, we further demonstrate that the same exponential bound applies to the convergence toward the exact optimal solution in some specific cases. 
 
Our contributions in a more technical perspective are summarized as follows:
\begin{itemize}
    \item First of all, the sequence of probability distributions generated by the Arimoto-Blahut algorithm converges to an $\varepsilon$-optimal solution for any sufficiently small constant $\varepsilon > 0$. As long as the current iterate is not an $\varepsilon$-optimal solution, the convergence rate toward such a solution is inverse exponential (geometric), given by $O\left({\log(m)}/{c^t} \right)$ for some constant $c>1$ (Theorem \ref{thm:rate of convergence}). 
       
    Importantly, the constant $c$ depends only on the constant approximation $\varepsilon$ and the problem size $m$, and not on specific properties of the channel. This stands in contrast to other geometric convergence results in the literature, such as those in \cite{Arimoto1972AnAF,Nakagawa2020AnalysisOT}, where the convergence rate may depend on channel-specific quantities, for example, the eigenvalues of a Jacobian matrix derived from the channel.
    \item If the region of the optimal solutions has a strictly positive constant volume, then the same convergence bounds apply for achieving an optimal solution (Theorem \ref{thm: converge to exact solution in the interior}).
    \item We complement the results of \cite{Arimoto1972AnAF,Nakagawa2020AnalysisOT} (Theorem 3 and Theorem 5, respectively) for geometric rates under specific conditions, giving global geometric rates (Theorem \ref{thm:new AB} and \ref{thm:new Nakagawa}, respectively). 
\end{itemize}

These contributions represent a significant improvement not only over the previously established upper bounds of $O\left({\log(m)}/{t}\right)$ for the convergence rate of Arimoto-Blahut algorithm, but also compared to the bounds of $O\left({\sqrt{\log(m)}}/{t}\right)$, of the algorithm derived in \cite{SutterSEL15}. 
To the best of our knowledge, the Arimoto-Blahut algorithm remains a state of art algorithm of this problem, although the problem of computing the capacity is a convex programming problem, so maybe classical methods of this can potentially be applied. 

\subsection{Further Related work}

Different versions of the Arimoto-Blahut algorithm are studied in the bibliography. Namely, authors in \cite{Sayir2000IteratingTA} studied the case where a small number of the input symbols are assigned non-zero probabilities in the capacity-achieving distribution, but no theoretical results regarding the speed up and the convergence rate are provided. In \cite{Yu2009SqueezingTA}, accelerated versions of the classical Arimoto-Blahut algorithm are provided by parametrizing the recurrence formula of the classical algorithm, and their work demonstrated that the convergence is at least as fast as the Arimoto-Blahut algorithm \cite{Arimoto1972AnAF}. Moreover, authors in \cite{Chen2023ACB} introduced a modification to the classical Arimoto-Blahut algorithm and examined the dual problem of computing the channel capacity, known as the rate-distortion problem. They reported a convergence rate of at least $O(1/t)$ and derived that their algorithm requires at least ${2\log(n)}/{\varepsilon}$ iterations. This results in $O\left(\frac{mn \log(n)}{\varepsilon} \left(1 + \log|\log(\varepsilon)|\right)\right)$ arithmetic operations to achieve an $\varepsilon$-optimal solution. Further, in \cite{Matz2004InformationGF} the authors provided a modified version of the Arimoto-Blahut algorithm utilizing natural gradient and proximal point methods, and they proved the convergence of their algorithm, but their convergence rate analysis is restricted only to the worst case, which is $O(1/t)$. There have been many variants of the Arimoto-Blahut algorithm introduced over the years, such as \cite{Vontobel2008AGO}, without improving the convergence rates of the vanilla algorithm.

\section{Preliminaries}

We consider a discrete memoryless channel $X\to Y$, where $X$ is the input of size $m$ and $Y$ of size $n$ is the output random variables, with $\{x_1,x_2,\dots,x_m\}$ being the possible symbols of the input and $\{y_1,y_2,\dots,y_n\}$ being the possible symbols of the output. In this paper, we consider that $m,n$ are arbitrarily large.

Given a probability distribution, column vector, $p \in \Delta^{m-1}$, where $\Delta^{m-1}$ is the $m-1$-dimensional simplex space, the support of $p$ is the set of indexes with strictly positive probability. Furthermore, we consider the discrete channel matrix $W$ with $w_{ij} \in [0,1]$, where $w_{ij}$ is the conditional probability of the $j$ output symbol when the input symbol $i$ is transmitted, with $\sum_{j \in [n]} w_{ij} =1$ for any possible $i \in [m] = \{1,2,\dots, m\}$. Thus, for any probability distribution $p$ of the input, the output probability distribution is $q^{\top} = p^{\top}\cdot W\in \Delta^{n-1}$. 

Further, the Shannon entropy \cite{Shannon48}, for a distribution $p$ is defined as $H(p) = - \sum_{i \in [n]} p_i \cdot \log(p_i)$, the mutual information given the channel $W$ and an input probability $p$ is defined as 
\begin{equation*}\label{eq:mutual information}
    I(p, W) = \sum_{i \in [m]} \sum_{j \in [n]} p_i \cdot w_{ij} \log\left(\frac{w_{ij}}{\sum_{k \in [m]} p_k \cdot w_{kj}}\right),
\end{equation*}
and the channel capacity $C^*$ is defined as $C^* = \max_{p \in \Delta^{m-1}} I(p, W)$. In other words, for any given input probability, we can calculate the mutual information that is no greater than the maximum capacity $C^*$ that is achieved by an input probability $p^*$. In general, this optimal distribution is not unique, and as was proven in \cite{Arimoto1972AnAF}, any convex combination of optimal solutions is also an optimal solution. Thus, the set of optimal solutions forms a convex subset of $\Delta^{m-1}$.

The \emph{Kullback-Leibler} (KL) divergence $D_{KL}(p || q)$ for two distributions $p, q \in \Delta^{m-1}$ is defined by $D_{KL}(p || q) = \sum_{i \in [m]} p_i \log\Big(\tfrac{p_i}{q_i}\Big)$, with the convention that $p_i \log\Big(\tfrac{p_i}{q_i}\Big) = 0$ whenever $p_i = 0$. 

\noindent By the necessary and sufficient KKT conditions for an optimal input distribution $p^* = (p^*_1, \ldots , p^*_m ) \in \Delta^{m-1}$ that achieves the capacity of the channel, see \cite{Jelinek1968ProbabilisticIT}, we have that:
\begin{equation*}
    D(W^i||q^*) \left\{ \begin{array}{ll}
        = C^*, & p^*_i>0, \\
        \leq C^*, & p^*_i=0,\\
    \end{array} \right.
\end{equation*}
where $(q^*)^{\top}= (p^*)^{\top}\cdot W$ and $W^i$ denotes the $i$-th row of the channel matrix $W$. In other words, for any index $i$ in the support of the optimal solution $p^*$, then $D(W^i||q^*) = C^*$. However, it is possible for an index $i$ that is not in the support of $p^*$ to achieve the capacity too. Based on this characterization, the authors in \cite{Nakagawa2020AnalysisOT} classify the indexes $[m]$ into three types:

\begin{itemize}
\item Type I: For $i\in [m]$ such that $D(W^i||q^*) = C^*$ and $p^*_i>0$,
\item Type II: For $i\in [m]$ such that $D(W^i||q^*) = C^*$ and $p^*_i = 0$,
\item Type III: For $i\in [m]$ such that $D(W^i||q^*) < C^*$ and $p^*_i = 0$.
\end{itemize}
For notational convenience, we denote as $m_I$, $m_{II}$, and $m_{III}$ the set of indexes of Type I, Type II, and Type III, respectively, at an optimal solution $p^*$. 

Finally, given a set $S$ we denote as $\interior S$ and $\partial S$ the interior and the boundary sets of $S$, respectively. We define the $\delta$-ball around a $p \in S$ as $B_S(p, \delta) := \{q \in S \mid \|p - q\|_1 \leq \delta\}$ (or simply as $B(p, \delta)$ when the set is obvious from the context). 

The computation of the capacity of a discrete memoryless channel is given by the seminal Arimoto-Blahut algorithm \cite{Arimoto1972AnAF,Blahut1972ComputationOC}, which we describe in following section.

\subsection{Description of the Arimoto-Blahut algorithm}

At the original Arimoto-Blahut algorithm, the time is split into discrete slots $t \in \mathbb{N}_0$. In \cite{Arimoto1972AnAF}, at any iteration $t+1$, the approximate channel capacity, $C(t+1, t)$, can be calculated, given the probability distribution $p^{t}$ from the previous iteration. Formally, this can be written as follows
\begin{equation}\label{eq: capacity at t and t - 1}
\begin{split}
    C(t+1,t) = - \sum_{i \in [m]} p^{t+1}_i \cdot \log(p^{t+1}_i) + \sum_{i \in [m]} \sum_{j \in [n]}  p^{t+1}_i \cdot w_{ij} \cdot \log \left( \frac{p^{t}_i \cdot w_{ij}}{\sum_{k \in [n]} p^{t}_k \cdot w_{kj} }\right).
\end{split}
\end{equation}
At each iteration $t+1$, in order to compute the $C(t+1,t)$, the new probability distribution $p^{t+1}$, given the previous distribution $p^{t}$, is computed using the following recurrence formula, see \cite{Arimoto1972AnAF,Nakagawa2020AnalysisOT},
\begin{equation}\label{eq:Arimoto-Blahut}
\begin{split}
    p^{t+1}_i = \frac{p^{t}_i \cdot e^{\sum_{j \in [n]} w_{ij} \cdot  \log\left(\frac{w_{ij}}{\sum_{k \in [m]} p^{t}_k \cdot w_{kj}}\right)}}{\mathlarger{\sum}\nolimits_{\ell \in [m]} p^{t}_{\ell} \cdot e^{\sum_{j \in [n]} w_{\ell j} \cdot  \log\left(\frac{w_{\ell j} }{\sum_{k \in [m]} p^{t}_k \cdot w_{kj}}\right)}}, \quad \text{for any } i \in [m],
\end{split}
\end{equation}
where $p^0$ is a full support distribution.

\subsection{Approximation of the capacity}\label{sec:approximate Arimoto Blahut}

Similar to \cite{Arimoto1972AnAF}, let the function $f(p):\Delta^{m-1} \to [0,C^*]$ with $f(p^{t}) = C^* - C(t+1, t)$, representing the distance/error from the capacity at iteration $t+1$, such that the problem of the computation of the capacity of the channel becomes a minimization problem of the function $f$, where the minimum value is $f(p^*)=0$, for any optimal distribution $p^*$. Now, we have the following approximation of the function $f$. 

\begin{definition}[Approximate optimum solution]
An input distribution $p$ is an $\varepsilon$-optimum solution, for any $\varepsilon\geq 0$, if and only if,
\begin{equation*}
    f(p)\leq \varepsilon.
\end{equation*}
\end{definition} 
\noindent It is easy to see that when $\varepsilon = 0$, then we have an (exact) optimum solution.
 Moreover, in this paper, we define two types of convergence rates for the function $f$:

\begin{definition}[Inverse linear rate]
The function $f$ has inverse linear rate of convergence, as $f(p^t) \to \varepsilon$ from the above, for any $\varepsilon\in [0,C^*]$, if and only if,
\begin{equation*}
    f(p^t) = O(1/t),
\end{equation*}
for any $t \in \mathbb{N}$.
\end{definition}

\begin{definition}[Inverse exponential rate]
The function $f$ has inverse exponential rate of convergence, as $f(p^t) \to \varepsilon$ from the above, for any $\varepsilon\in [0,C^*]$, if and only if,
\begin{equation*}
    f(p^t) = O(1/c^t),
\end{equation*}
for any $t \in \mathbb{N}$, with $c>1$.
\end{definition} 

\noindent Given these definitions, the previous main results from the bibliography on the convergence rate of the Arimoto-Blahut algorithm can be summarized as follows:

\begin{corollary}\label{cor:Arimoto-Blahut rate of convergence} [Corollary 1 in \cite{Arimoto1972AnAF}] The Arimoto-Blahut algorithm has rate of convergence for $f(p^t)$ at most $O\Big({\log(m)}/{t}\Big)$ to converge to an optimal solution.
\end{corollary}

\noindent This corollary implies that the Arimoto-Blahut algorithm requires at most $O\left({\log(m)}/{\varepsilon} \right)$ iterations to reach an $\varepsilon$-optimal solution \cite{SutterSEL15,Jurgensen1984ANO}.\\

\begin{theorem}\label{thm:Unique solution in interior} [ Theorem 3 in \cite{Arimoto1972AnAF}] Consider a unique optimal solution in the interior of the simplex, then there is an integer $N\geq 0$ s.t. the Arimoto-Blahut algorithm has inverse exponential rate of convergence to achieve this optimal solution, after $t \geq N$.\\
\end{theorem}

\begin{theorem}\label{thm:Nakagawa2020 m_II empty} [Theorem 5 in \cite{Nakagawa2020AnalysisOT}] Consider a unique optimal solution with $m_{II} = \emptyset$, then there is a $\delta>0$ s.t. for arbitrary initial distribution $p^0$, with $\| p^0- p^*\|\footnote{For some $r$-norm with $r \in \mathbb{N}$.}<\delta$,
the Arimoto-Blahut algorithm has inverse exponential rate of convergence\footnote{For the capacity or the distance/norm of the current distribution from the optimum. We can consider these two metrics as "equivalent".}
\begin{equation*}
    \|p^t - p^*\| < K \cdot \theta^t,
\end{equation*}
to achieve this optimal solution, for any $t \geq0$, with $K>0$ and $0\leq \theta <1$.
\end{theorem}
\noindent Note that, to our understanding, this result implies that the initial distribution $p^0$ should be arbitrarily close to the optimal solution, since there is no explicit guarantee on the value of $\delta$, which may be arbitrarily small. To address this limitation, in Section \ref{Exact solutions}, we provide a complementary analysis. In particular, building on Theorem \ref{thm:Nakagawa2020 m_II empty} (corresponding to Theorem 5 in \cite{Nakagawa2020AnalysisOT}), we establish a global inverse exponential rate of convergence under this setting. Finally, recently there was the following result.

\begin{theorem}\label{thm:Nakagawa2020} [Theorem 9 in \cite{Nakagawa2020AnalysisOT}] If $m_{II} \neq \emptyset$ in the unique optimal solution, then there is an initial distribution $p^0$ with $p^0_i>0$, for any $i\in m_{II}$, s.t. the Arimoto-Blahut algorithm has inverse linear rate of convergence to this optimal solution as $t \to \infty$.
\end{theorem}  

\section{Approximation analysis of the Arimoto-Blahut algorithm}\label{Approximation analysis}

In our paper, we study the rate of convergence of the function $f$ until the algorithm achieves an $\varepsilon$-optimal solution, for an $\varepsilon$ sufficiently small constant. To do this, first we have the following straightforward observations.

\begin{lemma}
For any optimum solution $p^*$ and any sufficiently small $\varepsilon>0$, there is a $\delta(\varepsilon)$ and a set of probability distributions $S \in \interior \Delta^{m - 1}$, s.t.  for any $p \in S \subseteq B(p^*,\delta(\varepsilon))$ it holds that $f(p) \leq \varepsilon$.
\end{lemma}
\begin{proof}
The proof is straightforward by the continuity of the function $f$.

\end{proof}

Denoting the interior of the set of $\varepsilon$-optimal solutions as $\interior B(p^*,\delta(\varepsilon))$, for any $p^*$ optimum solution, the previous Lemma implies that for any optimal solution $p^*$, even for solutions on the boundary $\partial \Delta^{m-1}$, that is with $m_{II}\neq \emptyset$, there is a set of approximate solutions $\interior B(p^*,\delta(\varepsilon))$ with $m_{II}= \emptyset$ and $f(p) \leq \epsilon$ for any $p\in \interior B(p^*,\delta(\varepsilon))$.

Let an arbitrary optimal solution $p^*$ from the convex set of the optimal solutions and $p^t$ the sequence of the input distributions that are generated by the Arimoto-Blahut algorithm \eqref{eq:Arimoto-Blahut}. The KL divergence between the optimum solution $p^*$ and $p^t$ is the following
\begin{equation*}
\begin{split}
D_{KL}(p^* || p^t) = \sum_{i \in [m]} p^*_i \log\Big(\tfrac{p^*_i}{p^t_i}\Big).
\end{split}
\end{equation*}

Then, the difference of the KL divergences between $p^*$ and two consecutive probability distributions $p^t$ and $p^{t+1}$, in iterations $t$ and $t+1$, respectively, is given by
\begin{equation}\label{eq:KL divergence difference}
    D_{KL}(p^* || p^{t+1}) - D_{KL}(p^* || p^t) = \sum\nolimits_{i \in [m]} p^*_i \log\left(\frac{p^t_i}{p^{t+1}_i}\right), \quad \text{for } t \geq 0.
\end{equation}

\begin{lemma}
\label{lemma:Difference of KL- f}
Consider the sequence of distributions $\{p^t\}_{t \in \mathbb{N_0}}$ generated by the Arimoto-Blahut algorithm, as described at equation \eqref{eq:Arimoto-Blahut}, starting from a full support uniform initial distribution $p^0$. Then, it holds that
\begin{equation}\label{Difference of KL- f}
\begin{split}
D_{KL}(p^*||p^{t+1}) - D_{KL}(p^*||p^{t}) = - f(p^t)- D_{KL}(q^*||q^t).
\end{split}
\end{equation}
\end{lemma}
\begin{proof}
The proof of this Lemma closely follows the approach used in equations (30) - (32) in the proof of Theorem 1 in \cite{Arimoto1972AnAF}, but we present it again for the sake of completeness of the analysis in our context.
 Thus, we have that
\begin{equation*}\label{eq:thm3eq1}
    D_{KL}(p^*||p^{t+1}) - D_{KL}(p^*||p^{t}) = \sum\nolimits_{i \in [m]} p^*_i \log\left(\frac{p^t_i}{p^{t+1}_i}\right).
\end{equation*}
Note here that since the initial probability has full support and the update rule, that is equation (\ref{eq:Arimoto-Blahut}), of the Algorithm is multiplicative keeping $p^t_i$ always positive for any $i$. In other words, the sequence generated by the Algorithm belongs to $\interior \Delta^{m-1}$.\\

\noindent Substituting now the Arimoto–Blahut update for $p^{t+1}_i$, the right-hand side becomes,
\begin{equation*}
    \sum_{i \in [m]} p^*_i \cdot  \log\left(\frac{p^t_i \cdot \mathlarger{\sum}\nolimits_{\ell \in [m]} e^{\sum_{j \in [n]} w_{\ell j} \cdot  \log\left(\frac{p^t_{\ell} \cdot w_{\ell j} }{\sum_{k \in [m]} p^t_k \cdot w_{kj}}\right)}}{p^t_i \cdot e^{ \sum_{j \in [n]} w_{ij} \cdot  \log\left(\frac{w_{ij}}{\sum_{k \in [m]} p^t_k \cdot w_{kj}}\right)}}\right).
\end{equation*}
Rearranging terms and using logarithmic properties we have,
\begin{equation*}
    \log\left(\mathlarger{\sum}_{\ell \in [m]} e^{\sum_{j \in [n]} w_{\ell j} \cdot  \log\left(\frac{p^t_{\ell} \cdot w_{\ell j} }{\sum_{k \in [m]} p^t_k \cdot w_{kj}}\right)}\right) - \sum_{i \in [m]} p^*_i \cdot  \sum_{j \in [n]} w_{ij} \cdot  \log\left(\frac{w_{ij}}{\sum_{k \in [m]} p^t_k \cdot w_{kj}}\right).
\end{equation*}
Setting $r^t_i = e^{\sum_{j \in [n]} w_{ij} \cdot  \log\left(\frac{p^t_{i} \cdot w_{ij} }{\sum_{k \in [m]} p^t_k \cdot w_{k j}}\right)}$ we take,
\begin{equation*}
    D_{KL}(p^*||p^{t+1}) - D_{KL}(p^*||p^{t})
    = \log\left(\sum_{i \in [m]} r^t_i\right) - \sum_{i \in [m]} p^*_i \cdot  \sum_{j \in [n]} w_{ij} \cdot  \log\left(\frac{w_{ij}}{\sum_{k \in [m]} p^t_k \cdot w_{kj}}\right).
\end{equation*}
Using that $C(t+1,t) = \log \Big(\sum\nolimits_{i \in [m]} r^t_i \Big)$ (see Corollary 1 and equation (28) in \cite{Arimoto1972AnAF}) and rearranging terms, the KL divergences difference becomes
\begin{equation*}
\begin{split}
    D_{KL}(p^* ||p^{t+1}) & - D_{KL}(p^*||p^{t}) = C(t+1,t) - \sum_{i \in [m]} p^*_i \cdot  \sum_{j \in [n]} w_{ij} \cdot  \log\left(\frac{w_{ij}}{\sum_{k \in [m]} p^t_k \cdot w_{kj}}\right)\\
    &= C(t+1,t) - \sum_{i \in [m]} p^*_i \cdot  \sum_{j \in [n]} w_{ij} \cdot  \log\left(\frac{w_{ij} \cdot \sum_{k \in [m]} p^*_k \cdot w_{kj}}{\left(\sum_{k \in [m]} p^t_k \cdot w_{kj}\right)\cdot \left(\sum_{k \in [m]} p^*_k \cdot w_{kj}\right)}\right)\\ 
    & = C(t+1,t) - \sum_{i \in [m]} p^*_i \cdot  \sum_{j \in [n]} w_{ij} \cdot  \log\left(\frac{w_{ij}}{\sum_{k \in [m]} p^*_k \cdot w_{kj}}\right) \\
    & \hspace{13em} - \sum_{i \in [m]} p^*_i \cdot  \sum_{j \in [n]} w_{ij} \cdot  \log\left(\frac{ \sum_{k \in [m]} p^*_k \cdot w_{kj}}{\sum_{k \in [m]} p^t_k \cdot w_{kj}}\right).\\
\end{split}
\end{equation*} Now, since $C^* = \sum_{i \in [m]} p^*_i \cdot  \sum_{j \in [n]} w_{ij} \cdot  \log\left(\frac{w_{ij}}{\sum_{k \in [m]} p^*_k \cdot w_{kj}}\right)$ we have that

\begin{equation*}
    D_{KL} (p^*||p^{t+1}) - D_{KL}(p^*||p^{t})  = C(t+1,t) - C^* - \sum_{i \in [m]}  \sum_{j \in [n]} p^*_i \cdot  w_{ij} \cdot  \log \left(\frac{\sum_{k \in [m]} p^*_k \cdot w_{kj}}{\sum_{k \in [m]} p^t_k \cdot w_{kj}} \right).\\
\end{equation*}
The output of the communication channel is $q^t_j = \sum_{k \in [m]} p^t_k \cdot w_{kj}$, hence
\begin{equation*}
    \sum_{i \in [m]}  \sum_{j \in [n]} p^*_i \cdot  w_{ij} \cdot  \log \left(\frac{\sum_{k \in [m]} p^*_k \cdot w_{kj}}{\sum_{k \in [m]} p^t_k \cdot w_{kj}} \right) = \sum_{j \in [n]}  \sum_{i \in [m]} p^*_i \cdot  w_{ij} \cdot  \log \left(\frac{\sum_{k \in [m]} p^*_k \cdot w_{kj}}{\sum_{k \in [m]} p^t_k \cdot w_{kj}} \right) = D_{KL}(q^*||q^t).
\end{equation*}
Therefore,
\begin{equation*}
    D_{KL}(p^*||p^{t+1}) - D_{KL}(p^*||p^{t}) = C(t+1,t) - C^* - D_{KL}(q^*||q^t)  = - f(p^t)- D_{KL}(q^*||q^t).
\end{equation*}
\end{proof}

\begin{theorem}
\label{Difference of KL}
Let an arbitrary optimal solution $p^*$ and the sequence of distributions $\{p^t\}_{t \in \mathbb{N}}$ generated by the Arimoto-Blahut algorithm, starting from a full support uniform initial distribution $p^0$. Then, for each iteration $t$, if the $p^{t}$ has not yet resulted in an $\varepsilon$-optimal solution, the KL divergence $D_{KL}(p^*||p^{t})$ strictly decreases by at least $\varepsilon$ in the next step, where $\varepsilon$ is a sufficiently small positive constant.
\end{theorem}

\begin{proof}
Since we analyze the convergence to an $\varepsilon$-approximate optimal solution, we can pick an arbitrary optimal solution $p^*$ from the convex set of the optimal solutions to prove that the algorithm always converges to this set that is an $\varepsilon$-optimum solution. \\

\noindent By Lemma \ref{lemma:Difference of KL- f}, in each iteration of the algorithm we have that
\begin{equation*}
    D_{KL}(p^*||p^{t+1}) - D_{KL}(p^*||p^{t}) = - f(p^t)- D_{KL}(q^*||q^t).
\end{equation*}
Since the KL divergence is non-negative, $D_{KL}(q^*||q^t) \geq 0$, it follows that
\begin{equation*}
    D_{KL}(p^*||p^{t+1}) - D_{KL}(p^*||p^{t}) \leq - f(p^t).
\end{equation*}
Rearranging terms we take a bound for the $f(p^t)$, that is 
\begin{equation}
\label{bound of f}
f(p^t) \leq D_{KL}(p^*||p^{t}) -D_{KL}(p^*||p^{t+1}).
\end{equation}
Now suppose that at iteration $t$, the algorithm provides a $p^t$ such that $f(p^t) > \varepsilon$, we obtain
 \begin{equation}
\label{KL inequality}
\begin{split}
    D_{KL}(p^*||p^{t+1}) < D_{KL}(p^*||p^{t}) -\varepsilon.
\end{split}
 \end{equation}
Thus, the KL divergence decreases by at least a constant quantity $\varepsilon$ in each iteration $t$, and the sequence $\{p^t\}_{t \in \mathbb{N}}$ approaches the convex set of the optimal solutions from outside.
\end{proof}

\noindent We can easily see by Theorem \ref{Difference of KL} that $D_{KL}(p^*||p^{t}) - D_{KL}(p^*||p^{t+1}) > \varepsilon$, thus $D_{KL}(p^*||p^{0}) - D_{KL}(p^*||p^{t+1}) > t\cdot \varepsilon$, for any iteration $t$ such that $f(p^t)>\varepsilon$. Therefore, we can take that $D_{KL}(p^*||p^{0})> t\cdot \varepsilon$, which implies that $t < \frac{D_{KL}(p^*||p^{0})}{\varepsilon} < \frac{\log(m)}{\varepsilon}$ matching the previous upper bound (see Corollary 1 in \cite{Arimoto1972AnAF} and \cite{SutterSEL15}) of the analysis for the Arimoto-Blahut algorithm to reach an $\varepsilon$-optimal solution.\\ 

Having the previous Theorem, we improve the analysis of the bound of the rate of convergence to achieve an $\varepsilon$-optimal solution, and this is the main result of our paper.

\begin{theorem}\label{thm:rate of convergence}
Let $p^0 \in \Delta^{m-1}$ be the full support uniform initial distribution. Then, the rate of convergence of the Arimoto-Blahut algorithm is inverse exponential $O\left(\frac{\log(m)}{c^t} \right)$, with $c>1$ be a constant, until the algorithm reaches an $\varepsilon$-optimal solution, with $\varepsilon$ a sufficiently small constant\footnote{We can also consider $\varepsilon$ as a function of the size of the problem $(m,n)$, for instance, $\varepsilon=\frac{1}{poly(m,n)}$, since in the analysis of the complexity we can consider that $(m,n)\to \infty$, but for a specific realization of the problem we treat $(m,n)$ as fixed constants even if they are arbitrarily large.}. 
\end{theorem}

\begin{proof}
Let an arbitrary iteration $t^*$ such that $f(p^{t^*})> \varepsilon$, then by (\ref{KL inequality}), for any $t \in [0,t^*]$, we have that
 \begin{equation}
\begin{split}
D_{KL}(p^*||p^{t+1}) < D_{KL}(p^*||p^{t}) -\varepsilon.
\end{split}
 \end{equation}

\noindent This implies that there is an $\eta>1$ such that
\[
    D_{KL}(p^* || p^{t+1}) =   D_{KL}(p^* || p^{t}) - \eta \cdot \varepsilon = (1 - b_{t}) \cdot D_{KL}(p^* || p^{t}),
\] 
with $b_{t} =  \frac{\eta \cdot \varepsilon} {D_{KL}(p^* || p^{t})}\geq \frac{\varepsilon} {\log (m) }$. 
So we can argue that, since both $\varepsilon$ and $\log (m)$ are constants, it holds $b_{t} \geq \zeta_{t}$, where $\zeta_{t}$ is a positive constant independent of $t$. This implies that for any iteration $t$, we have that
\begin{equation*}
\begin{split}
    D_{KL}(p^* || p^{t+1}) = (1-b_{t-1})\cdot D_{KL}(p^* || p^{t}) \leq (1-\zeta_{t-1}) \cdot D_{KL}(p^* || p^{t})\leq (1-\zeta) \cdot D_{KL}(p^* || p^{t}),
\end{split}
\end{equation*}
where $\zeta = \min_{t < t^*} \{\zeta_t\} \in (0,  1)$. Thus, this inequality recursively for any previous iteration until the initial implies that 
\begin{equation*}
\begin{split}
    D_{KL}(p^* || p^{t+1}) \leq (1-\zeta)^t \cdot D_{KL}(p^* || p^{0}) \leq (1-\zeta)^t \cdot \log (m),  
\end{split}
\end{equation*}
since $D_{KL}(p^*||p^{0}) \leq \log (m)$.
Now, by (\ref{bound of f}) we have that 
\begin{equation*}
\begin{split}
    f(p^{t+1}) \leq D_{KL}(p^* || p^{t+1}) -D_{KL}(p^* || p^{t+2}) \leq (1-\zeta)^t \cdot \log (m)=\frac{\log(m)}{c^t},  
\end{split}
\end{equation*}
since $D_{KL}(p^* || p^{t+2}) \geq 0$ and
$c = 1/(1-\zeta) > 1$ be a constant. Note that this rate implies at most $O\Big(\log_c \left(\frac{\log(m)}{\varepsilon}\right) \Big)$ iterations and $O\Big(m n \cdot \log_c \left(\frac{\log(m)}{\varepsilon}\right) \Big)$ total complexity, but by our analysis we can see that at worst case $c = \frac{1}{1-\zeta}  = \frac{\log m}{\log m - \varepsilon}$, which is $m$-dependent.
\end{proof}

\begin{figure}
    \centering
    \subfigure[Convergence to $\varepsilon$-optimal solution.]{\label{fig:thm4}\includegraphics[scale=.45]{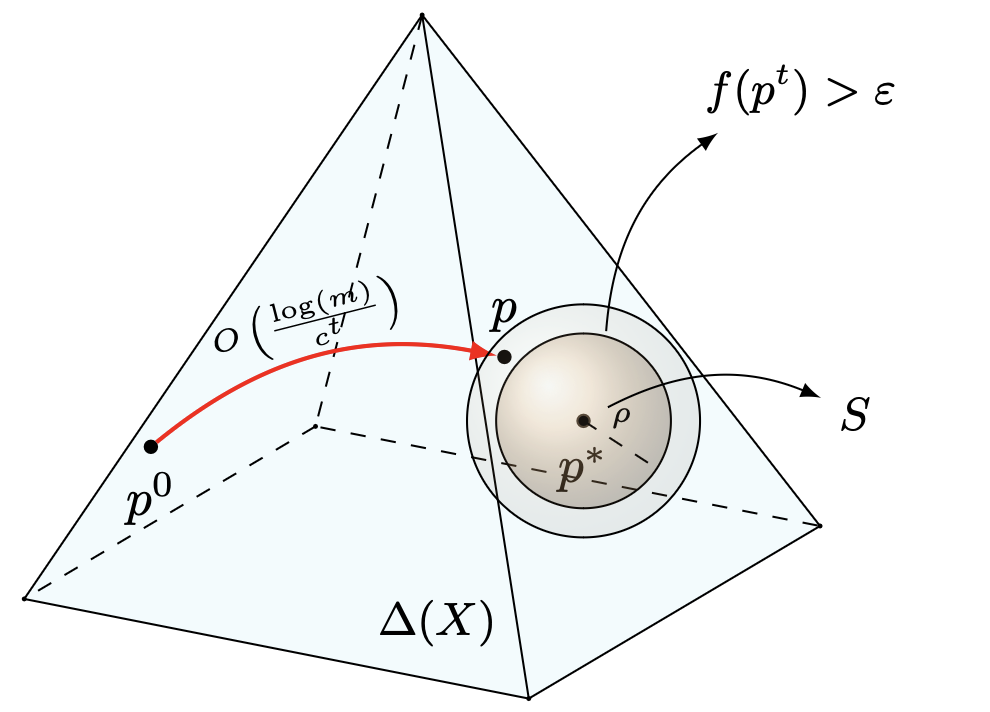}}
    \subfigure[Convergence to the interior of $S$.]{\label{fig:thm5}\includegraphics[scale=.45]{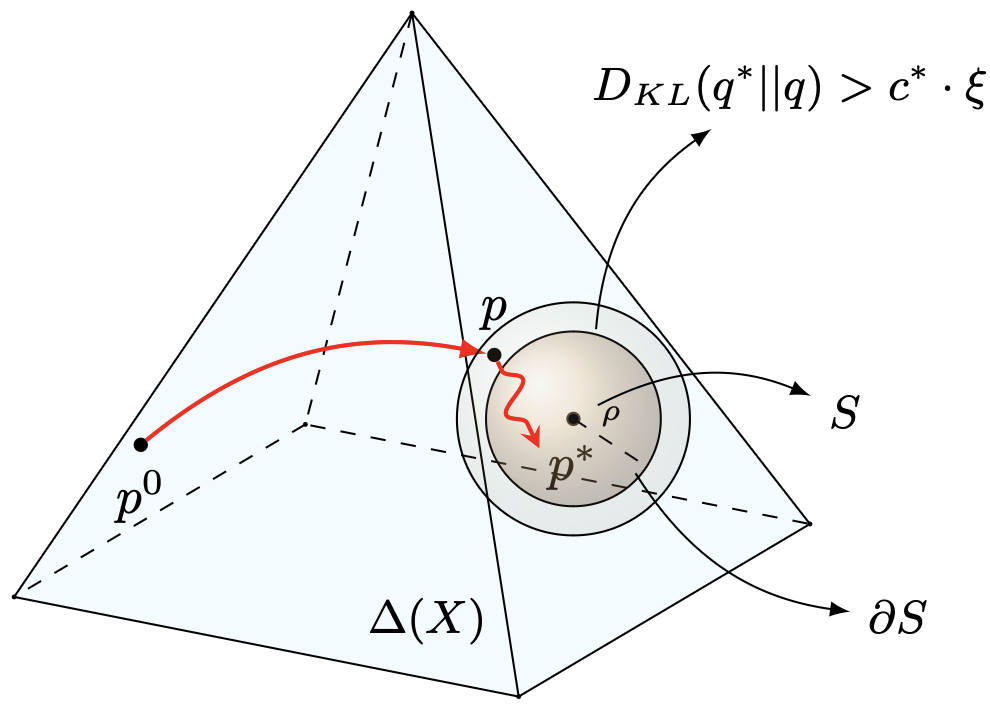}}
    \caption{Schematic representation for the convergence of the Arimoto-Blahut algorithm.}\label{fig:Convergence of AB}
\end{figure}

\section{Rate of convergence for exact solutions}
\label{Exact solutions}
We begin this section by recalling that an $m-1$-dimensional simplex $\Delta^{m-1}$ is the convex hull of $m$ vertices, and its interior consists of exactly those points that can be expressed as a convex combination of its vertices, where all coefficients are strictly between $0$ and $1$. Now, let $S$ be any convex subset of $\mathbb{R}^m$ that intersects the $\interior \Delta^{m-1}$. Therefore, $S$ contains points with strictly positive coefficients and its interior is non-empty. In fact, if $S$ is the set of capacity-achieving distributions, these are optimal solutions, there exists a $x_0 \in S$ and a radius $\rho > 0$ such that a maximum possible ball $B(x_0, \rho)$ is contained entirely in $S$, which means every point in that ball is an optimal solution.

Next, assuming that this radius $\rho$ is fixed, we will derive an explicit rate at which the Arimoto-Blahut algorithm's iterates converge to an exact optimal solution. Before doing so, we state and prove a few auxiliary properties that will be used in the subsequent analysis.

\begin{fact}
The KL divergence $D_{KL}(q^*|| q)$, with $(q^*)^{\top} = (p^*)^{\top} \cdot W$ and $q^{\top} = p^{\top} \cdot W$, is a convex function of $q$.
\end{fact}

\begin{lemma}
\label{lemma:boundary}
Consider the simplex $\Delta^{m-1}$ and the convex set of the optimal solutions $S$, with $B(p^*, \rho) \in S$ be the maximum possible ball in $S$, with center $p^*$ and constant radius $\rho$. Then, there exists a constant $\xi$ such that $\xi = \max \{D_{KL}(q^*|| q) \mid \text{ for any }p \in B(p^*, \rho)\}$, with $q^{\top} = p^{\top}\cdot W$, which is achieved on the boundary.  
\end{lemma}
\begin{proof}
The $B(p^*,\rho)$ is by definition a closed set, and since it is a subset of $\Delta^{m-1}$ simplex, it is also bounded. Therefore it is a compact set, and convex by construction.

Furthermore, since any $p \in B(p^*,\rho)$ is a full support distribution, any $q^{\top} = p^{\top}\cdot W$ is also a full support distribution. It follows that $D_{KL}(q^*|| q)$ is finite and continuous in $q$. 

A continuous function on a compact set attains its maximum, so $\xi = \max \{D_{KL}(q^*|| q) \mid \text{ for any }p \in B(p^*, \rho)\}$ exists. Finally, since $D_{KL}(q^*|| q)$ is a convex function, its maximum over the convex and compact set $B(p^*,\rho)$ occurs on the boundary. 
\end{proof}

\noindent By the previous Lemma, we have the following straightforward result.

\begin{corollary}
\label{corollary 1}
Under the assumptions of Lemma \ref{lemma:boundary}, if $D_{KL}(q^*|| q)<\xi$, for some $q,p$ s.t. $q^{\top} = p^{\top} \cdot W$, then $p$ is an optimal solution in $B(p^*,\rho)$.
\end{corollary}

\begin{theorem}\label{thm: converge to exact solution in the interior}
Consider the simplex $\Delta^{m-1}$ and the convex set of the optimal solutions $S$, with $B(p^*, \rho) \in S$ be the maximum ball in $S$, with center $p^*$ and constant radius $\rho$. Then, applying the Arimoto-Blahut algorithm with a full support uniform initial distribution $p^0$, the sequence $\{p^t\}$ converges to an exact optimal solution in $S$ with the rate of convergence equal to $O\left(\frac{\log(m)}{c^t} \right)$, for a constant $c > 1$. 
\end{theorem}

\begin{proof}
By the Lemma \ref{lemma:boundary} and Corollary \ref{corollary 1}, there exists a finite constant $\xi$ such that 
$\xi = \max \{D_{KL}(q^*|| q) \mid \text{ for any }p \in B(p^*, \rho)\}$, with $q^{\top} = p^{\top}\cdot W$. In particular, any solution $p^t$ such that $D_{KL}(q^*|| q^t) <  \xi$ is an optimal solution, so $f(p^t) = 0$. 

Now suppose that at iteration $t$, the algorithm provides a distribution $p^t$ that is not an optimal solution (i.e., $f(p^t) > 0$) and $D_{KL}(q^*|| q^t) \geq \xi$. Since $\xi > 0$, we can choose a sufficiently small constant $c^* \in (0,1)$ that $D_{KL}(q^*|| q^t) \geq \xi > c^* \cdot \xi$. \\

\noindent By the equality \eqref{Difference of KL- f} we have that 
\begin{equation*}
\begin{split}
D_{KL}(p^* || p^{t+1}) - D_{KL}(p^* || p^{t}) = - D_{KL}(q^*||q^t) - f(p^t).
\end{split}
\end{equation*}
Since $f(p^t) \geq 0$ and $D_{KL}(q^*|| q^t)> c^* \cdot \xi$, it follows that
\begin{equation*}
\begin{split}
D_{KL}(p^* || p^{t+1}) - D_{KL}(p^* || p^{t}) <- c^*\cdot \xi.
\end{split}
\end{equation*}
Picking $\eta > c^*$ and rearranging terms in the previous inequality we obtain
\begin{equation*}
\begin{split}
D_{KL}(p^* || p^{t+1})  =  D_{KL}(p^* || p^{t})- \eta \cdot \xi = (1-b_{t}) \cdot D_{KL}(p^* || p^{t}),
\end{split}
\end{equation*}
with $b_{t} =  \frac{\eta \xi} {D_{KL}(p^* || p^{t})} \geq \frac{c^* \xi} {\log (m) }$. 
So, similar to the proof of Theorem \ref{thm:rate of convergence} we can argue that, since $c^*,\xi$ and $\log (m)$ are constants, it holds $b_{t} \geq \zeta_{t}$, where $\zeta_{t}$ is a constant independent of $t$. This implies that for any iteration $t$, we have that
\begin{equation*}
\begin{split}
    D_{KL}(p^* || p^{t+1}) = (1-b_{t-1})\cdot D_{KL}(p^* || p^{t}) \leq (1-\zeta_{t-1}) \cdot D_{KL}(p^* || p^{t})\leq (1-\zeta) \cdot D_{KL}(p^* || p^{t}),
\end{split}
\end{equation*}
where $\zeta = \min_{t < t^*} \{\zeta_t\} \in (0,  1)$. Thus, by backward recursion we have
\begin{equation*}
\begin{split}
    D_{KL}(p^* || p^{t+1}) \leq (1-\zeta)^t \cdot D_{KL}(p^* || p^{0}) \leq (1-\zeta)^t \cdot \log (m),  
\end{split}
\end{equation*}
since $D_{KL}(p^*||p^{0}) \leq \log (m)$.
Finally, using the bound (\ref{bound of f}) and $D_{KL}(p^* || p^{t+2}) \geq 0$, we conclude
\begin{equation*}
\begin{split}
    f(p^{t+1}) \leq D_{KL}(p^* || p^{t+1}) -D_{KL}(p^* || p^{t+2}) \leq (1-\zeta)^t \cdot \log (m)=\frac{\log(m)}{c^t},  
\end{split}
\end{equation*}
where $c = 1/(1-\zeta) > 1$ is a constant. This establishes the claim.

This result implies that as far as we are not in an optimal solution $p^t$, then the sequence converges to an optimal solution within $S$ with inverse exponential rate. Finally, utilizing this inverse exponential bound, in a similar manner as in the proof of Theorem \ref{thm:rate of convergence},
we achieve the same bounds for the number of iterations, $O\left(\log_c \left(\frac{\log(m)}{\varepsilon}\right)\right)$, and total complexity, $O\left(m n \cdot \log_c \left(\frac{\log(m)}{\varepsilon}\right)\right)$, for some constant $\varepsilon>0$, to achieve a solution\footnote{Not necessary the solution $p^*$, since when the algorithm achieves an optimal solution it stops.} in $S$.
\end{proof}

\noindent The situation described in Theorem \ref{thm: converge to exact solution in the interior} is illustrated in Figure \ref{fig:thm5}.\\ 

By Theorem \ref{thm:Unique solution in interior} and Theorem \ref{thm:Nakagawa2020 m_II empty} an inverse rate of convergence for exact solutions is obtained, but only under rather restrictive conditions.
In contrast, the analysis of Section \ref{Approximation analysis} allows us to drop these conditions and still maintain exponential convergence. We begin with auxiliary corollaries.\\

\begin{corollary}[Implied by Theorem \ref{thm:rate of convergence}]
\label{Corollary:3}
For any $\delta>0$, there exists a sufficiently small constant $\varepsilon > 0$ and an iteration $t^0$ of the algorithm s.t. $D\left(p^* || p^{t^0}\right) < \delta$. Moreover, for any iteration $t < t^0$, the rate of convergence is inverse exponential $O\left(\frac{\log m}{c^t}\right)$, where $c > 1$ is a constant depending on $\varepsilon$ as provided by the Theorem \ref{thm:rate of convergence}.\\
\end{corollary}

\begin{corollary}[Implied by Theorem \ref{thm:rate of convergence}]
\label{Corollary:4}
For any $\delta > 0$, there exists a sufficiently small constant $\varepsilon > 0$ and an iteration $t^0$ of the algorithm s.t. $\|p^{t^0} -p^*\| < \delta$. Furthermore, for any iteration $t < t^0$, the rate of convergence is inverse exponential $O\left(\frac{\log m}{c^t}\right)$, where $c>1$ is a constant depending on $\varepsilon$ as provided by the proof\footnote{Note that Theorem \ref{thm:rate of convergence} quantifies convergence in terms of the KL divergence (and the decrement $f(p)$), but  by invoking Pinsker’s inequality one immediately recovers the same exponential rate in the distance norm between the current iterate and the optimum.} of Theorem \ref{thm:rate of convergence}.\\
\end{corollary}

\noindent Combining these and the results of Theorem \ref{thm:Unique solution in interior} and Theorem \ref{thm:Nakagawa2020 m_II empty} we have the following results.

\begin{theorem}\label{thm:new Nakagawa}
Consider a unique optimal solution $p^* \in \interior \Delta^{m - 1}$, then, starting from a full support uniform distribution, the Arimoto-Blahut algorithm has inverse exponential rate of convergence to achieve $p^*$.
\end{theorem}

\begin{proof}
By the Corollary \ref{Corollary:3}, for any $\delta>0$ in the proof of Theorem 3 in \cite{Arimoto1972AnAF}, there is sufficiently small constant an $\varepsilon$ s.t. $D\left(p^* || p^{t^0}\right)< \delta$, and for any $t < t^0$, the rate of convergence is inverse exponential. After iteration $t^0$ the proof of Theorem 3 in \cite{Arimoto1972AnAF} is applied, giving an inverse rate of convergence for any $t$ until we reach the optimal solution.

\end{proof}

\begin{theorem}\label{thm:new AB}
Consider a unique optimal solution $p^*$ with $m_{II} = \emptyset$, then, starting from a full support uniform distribution, the Arimoto-Blahut algorithm has inverse exponential rate of convergence to achieve $p^*$.
\end{theorem}

\begin{proof}
By the Corollary \ref{Corollary:4}, for any $\delta > 0$ in the proof of Theorem 5 in \cite{Nakagawa2020AnalysisOT}, there is sufficiently small constant an $\varepsilon$ s.t. $\| p^{t^0}-p^*\|< \delta$, and for any $t < t^0$, the rate of convergence is inverse exponential. After iteration $t^0$ the proof of Theorem 5 in \cite{Nakagawa2020AnalysisOT}, is applied, giving an inverse rate of convergence for any $t$ until we reach the optimal solution.
\end{proof}

\section{Discussion}\label{discussion}

In this paper, we revisited the classical Arimoto-Blahut algorithm regarding the channel capacity problem and introduced significant analysis improvements that confirm its computational efficiency. Our new results provide mathematical guarantees for achieving an $\varepsilon$-optimal solution, for any constant $\varepsilon>0$, in substantially fewer iterations than previously established in the literature. By scrutinizing the convergence analysis, we demonstrated a reduced overall complexity, indicating that the Arimoto-Blahut algorithm can indeed be practical for large-scale applications. These theoretical advancements contribute to a deeper understanding of the algorithm's performance and have implications for a range of applications where the Arimoto-Blahut algorithm, or similar algorithms are employed, such as in rate-distortion theory and the Expectation-Maximization algorithm.

Furthermore, our analysis improves upon the results in \cite{Nakagawa2020AnalysisOT} and \cite{SutterSEL15}. Specifically, our findings suggest that tighter Jacobian bounds may be obtained using Fisher information metric instead of $\ell_2$ for the former, and that adaptive smoothing should be introduced in the dual algorithm for the latter.

Finally, our work opens up several promising directions for future research. These include establishing lower bounds on the iteration complexity of the algorithm and extending the analysis to channels with continuous input alphabets. We believe that the techniques developed in this paper may also prove useful in the analysis of other algorithms in Information Theory and TCS.

\section*{Acknowledgments}

For this work the first and the second author were partially supported by the LASCON project of ICS-FORTH's internal grants 2022.

\printbibliography

@article{Arimoto1972AnAF,
  title={An algorithm for computing the capacity of arbitrary discrete memoryless channels},
  author={Suguru Arimoto},
  journal={{IEEE} Transactions on Information Theory},
  year={1972}
}

@article{Nakagawa2020AnalysisOT,
  title={Analysis of the Convergence Speed of the Arimoto-Blahut Algorithm by the Second-Order Recurrence Formula},
  author={Kenji Nakagawa and Yoshinori Takei and Shintaro Hara and Kohei Watabe},
  journal={IEEE Transactions on Information Theory},
  year={2020}
}

@article{Blahut1972ComputationOC,
  title={Computation of channel capacity and rate-distortion functions},
  author={Richard E. Blahut},
  journal={{IEEE} Transactions on Information Theory},
  year={1972}
}

@book{Jelinek1968ProbabilisticIT,
  title={Probabilistic Information Theory: Discrete and Memoryless Models},
  author={Frederick Jelinek},
  year={1968},
  publisher={McGraw-Hill}
}

@article{Shannon48,
  author       = {Claude E. Shannon},
  title        = {A mathematical theory of communication},
  journal      = {Bell Syst. Tech. J.},
  volume       = {27},
  number       = {3 \& 4},
  pages        = {(3): 379--423 \& (4): 623--656},
  year         = {1948}
}

@article{SutterSEL15,
  author       = {Tobias Sutter and
                  David Sutter and
                  Peyman Mohajerin Esfahani and
                  John Lygeros},
  title        = {Efficient Approximation of Channel Capacities},
  journal      = {{IEEE} Transactions on Information Theory},
  volume       = {61},
  number       = {4},
  pages        = {1649--1666},
  year         = {2015}
}

@article{Jurgensen1984ANO,
  title={A note on the Arimoto-Blahut algorithm for computing the capacity of discrete memoryless channels},
  author={H. Jurgensen},
  journal={IEEE Transactions on Information Theory},
  year={1984},
  volume={30},
  pages={376-377}
}

@article{Yu2009SqueezingTA,
  title={Squeezing the Arimoto–Blahut Algorithm for Faster Convergence},
  author={Yaming Yu},
  journal={IEEE Transactions on Information Theory},
  year={2009},
  volume={56},
  pages={3149-3157}
}

@article{Sayir2000IteratingTA,
  title={Iterating the Arimoto-Blahut algorithm for faster convergence},
  author={Jossy Sayir},
  journal={2000 IEEE International Symposium on Information Theory},
  year={2000},
  pages={235-}
}

@article{Matz2004InformationGF,
  title={Information geometric formulation and interpretation of accelerated Blahut-Arimoto-type algorithms},
  author={Gerald Matz and Pierre-Emile J. Duhamel},
  journal={Information Theory Workshop},
  year={2004},
  pages={66-70}
}

@inproceedings{Chen2023ACB,
  title={A Constrained BA Algorithm for Rate-Distortion and Distortion-Rate Functions},
  author={Lin Chen and Shitong Wu and Wen-Long Ye and Huihui Wu and Wen-Ying Zhang and Hao Wu and Bo Bai},
year={2024},
      eprint={2305.02650},
      archivePrefix={arXiv}
}

@book{Cover2006,
  author = {Cover, Thomas M. and Thomas, Joy A.},
  publisher = {Wiley-Interscience},
  title = {Elements of {I}nformation Theory (2nd Edition)},
  year = 2006
}

@article{Vontobel2008AGO,
  title={A Generalization of the {B}lahut–{A}rimoto Algorithm to Finite-State Channels},
  author={Pascal O. Vontobel and Aleksandar Kavcic and Dieter-Michael Arnold and Hans-Andrea Loeliger},
  journal={IEEE Transactions on Information Theory},
  year={2008},
  volume={54},
  pages={1887-1918}
}

\end{document}